\documentclass[11pt,a4paper]{amsart}
\usepackage{amsaddr}
\usepackage{graphicx}
\usepackage{tikz}
\usetikzlibrary{decorations.pathmorphing}
\usepackage{amsmath}
\usepackage{amssymb}
\usepackage{amsthm}
\usepackage{color}

\usepackage[numbers,sort&compress]{natbib}

\usepackage{hyperref}
\usepackage{bm}
\usepackage{dsfont}
\usepackage{multirow}

\usepackage{url}
\usepackage{graphicx,epsfig}
\usepackage{hyperref}
\usepackage{verbatim}
\usepackage{mathtools}
\usepackage{dsfont}
\usepackage{float}
\usepackage{amsmath,amsthm,amssymb,tabu,xparse}
\newtheorem{theorem}{Theorem}
\newtheorem*{theorem*}{Theorem}
\newtheorem{lemma}{Lemma}
\newtheorem*{lemma*}{Lemma}

\usepackage{dcolumn}
\usepackage{bm}

\newcommand{\ket}[1]{|{#1}\rangle}
\newcommand{\bra}[1]{\langle{#1}|}
\DeclarePairedDelimiter\abs{\lvert}{\rvert}%
\newcommand{\kron}[2]{\ket{#1}\bra{#2}}

\begin{document}
\title{An upper bound on the time required to implement unitary operations }
\author{Juneseo Lee$^*$, Christian Arenz$^*$, Daniel Burgarth$^\dagger$, and Herschel Rabitz$^*$} 
\address{$^*$Frick Laboratory, Princeton University, United States}
\address{$^\dagger$Department of Physics and Astronomy, Center of Engineered Quantum Systems, Macquarie University, Sydney, NSW 2109, Australia}
\date{\today}

\begin{abstract}
We derive an upper bound for the time needed to implement a generic unitary transformation in a $d$ dimensional quantum system using $d$ control fields. We show that given the ability to control the diagonal elements of the Hamiltonian, which allows for implementing any unitary transformation under the premise of controllability, the time $T$ needed is upper bounded by $T\leq \frac{\pi d^{2}(d-1)}{2g_{\text{min}}}$ where $g_{\text{min}}$ is the smallest coupling constant present in the system. We study the tightness of the bound by numerically investigating randomly generated systems, with specific focus on a system consisting of $d$ energy levels that interact in a tight-binding like manner.
\end{abstract}
\maketitle

\section{\label{sec:intro} Introduction}
Controlling quantum systems through external classical fields on a time scale that is below the typical decoherence timescale is crucial for employing quantum mechanical features for future quantum devices. In particular, the length $T$ of the classical pulses, sometimes referred to as the \emph{minimum gate time}, used to implement a target unitary transformation $U_{g}$ that allows to carry out a specific quantum information task should scale in a reasonable manner with the underlying Hilbert space dimension $d$. 
While lower bounds on $T$, known as \emph{quantum speed limits} (for a detailed review we refer to \cite{SpeedLimitsRev}, and works cited therein), give inherent limits on how fast unitary operations (states) can be implemented (prepared) through shaped classical fields, such lower bounds do not yield much insights on how much time is at most needed to achieve this task. Thus, an upper bound on $T$ is highly desirable. Such an upper bound should depend on the target unitary transformation, the Hamiltonian describing the quantum system under consideration, the number of controls available to implement the target transformation, and possible constraints, such as energy and bandwidth in the control fields. Clearly, if every matrix element of the Hamiltonian describing a $d$ dimensional quantum system can be controlled instantaneously and arbitrarily, every unitary transformation in the unitary group $\text{U}(d)$ can be implemented instantaneously through $d^{2}$ (unconstrained) classical fields controlling each matrix element. But what if we have only restricted access to the system under consideration? How many controls, and which controls, then  allow for implementing every $U_{g}\in \text{U}(d)$ in a time at most $\mathcal O(\text{poly}(d))$? Here we show that if the diagonal elements of the Hamiltonian describing a $d$ dimensional quantum system can be generically controlled through classical fields, and if the system is controllable with these fields, the time to implement every unitary operations scales at most as $\mathcal O(d^{3})$. We note, however, that for qubit systems consisting of $n$ qubits (i.e., $d=2^{n}$) our upper bound scales exponentially in $n$. This should not be surprising, as the time $T$ to implement a generic unitary transformation scales exponentially in the number of qubits, which can be traced back to the fact that most unitary operations cannot be implemented efficiently, i.e., in a time that scales polynomially in the number of qubits \cite{nielsen2002quantum}. For further reading regarding time-optimal control and quantum computing we refer to the seminal works \cite{khaneja2001time, nielsen2006quantum}, while an upper bound on $T$ for qubit systems was developed in \cite{arenz2018controlling}.

While in this work we mostly focus on networks determined by a set of basis states $\{\ket{n}\}$ describing a $d$ dimensional quantum system, we also consider the generalization to networks consisting of qubits. Here the associated graph is not determined by a coupling between two kets, but instead by qubits coupled through an arbitrary two-body interaction term. Based on the number of CNOT gates needed to create a specific unitary transformation \cite{barenco1995elementary, mottonen2004quantum, shende2006synthesis}, we thus also provide an upper bound on $T$ to implement a given $U_{g}$ on a $n$-qubit network using $2n$ local controls.  

 One way to obtain an upper bound on $T$ is to find a sequence of gates that corresponds to some application of the controls that allows for creating a generic unitary transformation. Upper bounding the corresponding time needed to implement the sequence then yields an upper bound for implementing a generic unitary transformation. This strategy has, for instance, been successfully applied for a $n$-qubit network to characterize the set of gates that can be implemented in a time at most polynomial in the number of qubits  using $2n$ local controls \cite{arenz2018controlling}. Here we build up on the concepts developed in \cite{arenz2018controlling} and show for a $d$ dimensional quantum system described by a Hamiltonian 
 \begin{align}
  \label{eq:Ham}
  H_{0}=\sum_{n,m}g_{n,m}|n\rangle\langle m|,
  \end{align}
that if the associated graph is connected, the set of controls $\mathcal C=\{|n\rangle\langle n|\}_{n=1}^{d}$ allows for implementing every $U_{g}\in\text{U}(d)$ in a time which is upper bounded by 
\begin{align}
\label{eq:upperbound}
T\leq \frac{\pi d^2 (d-1)}{2g_{\text{min}}} ,
\end{align}
where $g_{\text{min}}= \min_{n\neq m} \{\abs{g_{n,m}}\}$. Thus, fixing a basis to represent the Hamiltonian $H_{0}$ of the quantum system under consideration and controlling the diagonal elements of $H_{0}$ allows for implementing every unitary transformation in a time that is upper bounded by \eqref{eq:upperbound}, provided the system is fully controllable through the controls used. In fact we show that if the graph associated to \eqref{eq:Ham} is connected, the controls $\mathcal C$ generate a fully controllable system and the time $T$ to implement every unitary operation is upper bounded by \eqref{eq:upperbound}.

The work is organized as follows. We begin in section 2.1 by representing a quantum system described by the Hamiltonian \eqref{eq:Ham} as a weighted and undirected graph. We show in section 2.2 that we can reduce such a graph into a single-edge without a time cost using dynamical decoupling. In section 3.1 we present our first result.  Assuming the graph is connected, by propagating single-edge evolutions across a path connecting any pair of vertices, any  interaction between a pair of vertices can be created in a time linear in $d$. In section 3.2 we derive our main result \eqref{eq:upperbound}. The obtained results rely on the assumption that the control fields corresponding to the controls $\mathcal C$ are unconstrained, so that every unitary operation $v=\exp(-i\alpha C)$ with $C\in \mathcal C$ can be implemented instantaneously. In section 4 we numerically study the tightness of the derived upper bound by considering examples, followed by some econcluding remarks in section 5.  
\section{\label{sec:motiv} Preliminaries}
We consider a $d$ dimensional quantum control system  evolving on the unitary group $\text{U}(d)$ described by the Schr\"odinger equation for the time evolution operator
\begin{align}
\label{eq:controlsystem}
\frac{d}{dt}U(t)=-iH(t)U(t),
\end{align}
where we set $\hbar=1$ and the time dependent Hamiltonian is given by 
\begin{align}
\label{eq:totalham}
H(t)=\sum_{n\neq m}g_{n,m}|n\rangle\langle m|+\sum_{n}f_{n}(t)P_{n}.
\end{align}
We refer to
\begin{align}
\label{eq:driftham}
H_{0}=\sum_{n\neq m}g_{n,m}|n\rangle\langle m|,
\end{align}
as the drift Hamiltonian and the set of controls $\mathcal C=\{P_{n}\}_{n=1}^{d}$ is given by orthonormal projections $P_{n}=|n\rangle\langle n|$, where $f_{n}(t)$ are the corresponding control fields, which are throughout this work assumed to be unconstrained. Typically, the goal in quantum control is to shape the control fields in such a way that for some time $T$ the solution $U(T)$ to $\eqref{eq:controlsystem}$ is given by a desired unitary transformation $U_{g}$.  We call the control system fully controllable if every $U_{g}\in\text{U}(d)$ can be implemented through shaping the control fields. It is well known that the system is fully controllable iff the dynamical Lie algebra \cite{BookDalessandro, Elliot} generated by the drift Hamiltonian and the set of controls spans the full space, i.e., the algebra $\mathfrak{u}(d)$ consisting of $d\times d$ skew hermitian matrices. However, how much time $T$ does it take to implement a generic unitary transformation? 
In order to derive the upper bound \eqref{eq:upperbound}, we start by explaining how $H_{0}$ can be represented as a weighted, undirected graph, and how the controls $\mathcal C$ can be used to instantaneously remove vertices from the graph.  

\subsection{Quantum control systems and graphs}
We first note that we can rewrite the drift Hamiltonian as 
\begin{align}
H_{0}&=\sum_{n> m} \abs{g_{n,m}}(e^{i\phi_{n,m}}\kron{n}{m}+e^{-i\phi_{n,m}}\kron{m}{n}),
\end{align}
where the relative phases $\phi_{n,m}$ can be removed by applying the unitary transformations $v=\exp(-i\alpha C)$ with $C\in\mathcal C$, which can be implemented instantaneously assuming that the control fields $f_{n}(t)$ are unconstrained. We remark here that this assumption is a reasonable approximation in the case where the strength of the control fields can be made much larger than the typical energy scales of the systems being considered. Thus, in the case of unconstrained control fields removing the phases does not take time. Using the controls we can therefore map $H_{0}$ given by \eqref{eq:driftham} into the drift Hamiltonian 
\begin{align}
\label{eq:drifttrans}
H_{0}=\sum_{n>m}|g_{n,m}|B_{n,m},
\end{align}
without a time cost, where we defined $B_{n,m}=|n\rangle\langle m|+|m\rangle\langle n|$. In order to derive the bound \eqref{eq:upperbound} we can hence equivalently work with the drift Hamiltonian given by \eqref{eq:drifttrans}. 

The operators $B_{n,m}$ describe interactions between the states $\ket{n}$ and $\ket{m}$, whereas $\abs{g_{n,m}}$ is the corresponding interaction strength. We can visualize these interactions through a weighted and undirected graph $G(V,E)$. The set of vertices $V$ correspond to the basis states $\{\ket{n}\}$ spanning the Hilbert space of the quantum system and are labeled by $n$, the set of edges $E$ labeled by $(n,m)$ describe interactions between vertices $n$ and $m$, and the interaction strength $\abs{g_{n,m}}$ corresponds to the weights. Later on we will also consider qubit graphs in which $V$ and $E$ represent qubits and two-body interactions, respectively.

For a drift Hamiltonian $H_{0}$ of the form \eqref{eq:drifttrans} we denote the corresponding graph by $G_{0}(V_{0},E_{0})$. In order to upper bound the time to implement a generic target unitary transformation it is useful to introduce the complete graph $G_{K}(V_{K},E_{K})$ which consists of $|V_{K}|=d$ vertices and $|E_{K}|=\frac{d(d-1)}{2}$ edges. 

We proceed by first showing how to remove edges from $G_{0}$ instantaneously, followed by upper bounding the time to connect a generic pair of vertices, i.e., creating a generic $B_{n,m}$ with $(n,m)\in E_{K}$.

\subsection{Dynamical decoupling: removing edges without a time cost}

Dynamical decoupling allows for removing unwanted interactions of a Hamiltonian $H_{0}$ by rapidly applying a set of unitary transformations $V$ \cite{viola1999dynamical,lidar2013quantum} in a Suzuki-Trotter type sequence, 
\begin{align}
\label{eq:trotter}
\Lambda_{t/n}=\prod_{v\in V} v^{\dagger}\exp\left(-iH_{0}\frac{t}{|V|n}\right)v,
\end{align}
which converges in the limit of infinitely fast operations $(n\to\infty)$ to a unitary operation $U=\exp(-i M(H_{0})t)$, where the map $M$ is given by   
\begin{align}
\label{eq:maptrotter}
M(\cdot)=\frac{1}{|V|}\sum_{v\in V}v^{\dagger}(\cdot)v. 
\end{align}
We remark here that such maps are typically studied in the context of Hamiltonian simulation (see e.g.,\cite{bremner2005simulating}), dynamical decoupling (see e.g., \cite{viola1999dynamical,lidar2013quantum,arenz2017dynamical}), and unital quantum channels with equal weights \cite{nielsen2002quantum}. We further note that a concatenation $M_{1}(M_{2}(H_{0}))=\tilde{M}(H_{0})$ yields again a map of the form \eqref{eq:maptrotter}, so that $\tilde{M}$ can be obtained by a sequence of the form \eqref{eq:trotter}. 

Now, taking $V_{j}=\{\mathds{1},v_{j}\}$ where $v_{j}=\exp(-i\pi P_{j})=\mathds{1}-2P_{j}$ we have for $H_{0}$ given by \eqref{eq:drifttrans},  
\begin{align}
M_{j}(H_{0})=H_{0}-(P_{j}H_{0}+H_{0}P_{j}),
\end{align}
so that $M_{j}$ maps the graph $G_{0}$ into a subgraph in which the vertex $j$ is removed. Iteratively removing vertices by constructing concatenations of different $M_{j}$'s therefore allows to map $G_{0}$ into a subgraph containing only a single edge. We conclude that there always exists a set of unitary transformations generated by the controls $\mathcal C$ that allows for mapping the natural evolution given by $H_{0}$ into an evolution generated by $|g_{n,m}|B_{n,m}$ without a time cost  \cite{zeier2004gate}. Thus, for the control system defined by \eqref{eq:totalham} any unitary operation of the form 
\begin{align}
\label{eq:intunitary}
S_{n,m}(\alpha)=\exp(-i\alpha B_{n,m}),
\end{align} 
with $(n,m)\in E_{0}$ can be implemented through the controls in a time $t_{n,m}=\alpha/|g_{n,m}|$. We note that within the subspace spanned by $\{\ket{n},\ket{m}\}$ the operation $S_{n,m}(\alpha)$ induces oscillations between the states $\ket{n}$ and $\ket{m}$, and for $\alpha=\pi/2$ the states are swapped. Adopting the terminology used in quantum information, we refer to the corresponding unitary transformation $S_{n,m}(\pi/2)$ as a SWAP gate, noting that here two basis states are swapped rather than qubit states.   

\section{\label{sec:res1} Results}
So far we have shown that unitary operations generated by interactions $B_{n,m}$ of the Hamiltonian \eqref{eq:drifttrans} can be implemented in a time which is of the order of the inverse energy associacted with the interaction. However, how much time does it take to create interactions that are not present in $H_{0}$? In order to upper bound the time to implement a unitary operation generated by such interactions, we now show how to upper bound the time to create arbitrary interactions.        
\subsection{Upper bounding the time to create interactions}
From now on we assume that $G_{0}$ is connected and upper bound the time to create a generic $S_{n,m}$ with $(n,m)\in E_{K}$, where $E_{K}$ is the set of edges of the complete graph $G_{K}$. We establish the following lemma. 
\begin{lemma}\label{th:11}
Let the graph associated to the drift Hamiltonian \eqref{eq:drifttrans} be connected and denote by $g_{\text{min}}$ the smallest edge weight. Then for the control system \eqref{eq:controlsystem} a unitary operation $S_{n,m}(\alpha)\in \text{U}(d)$ with $(n,m)\in E_{K}$ of the form \eqref{eq:intunitary} can be implemented in a time $t_{n,m}$ which is upper bounded by 
\begin{align} 
\label{eq:upperboundtnm}
t_{n,m}\leq \frac{|\alpha|+\pi (d-2)}{g_{\text{min}}}.
\end{align}
\end{lemma}
\begin{proof}
We relabel the vertices of $G_{0}$ so that the vertices $n$ and $m$ are labeled by $1$ and $N$ and we consider a path connecting the vertices $1$ and $N$ by passing through connected vertices $1,2,\cdots,N$, i.e., $(j,j+1)\in E_{0}$ where $j=1,\cdots N-1$. Using dynamical decoupling we can instantaneously reduce $G_{0}$ to a single edge $(1,2)$, so that the SWAP gate $S_{1,2}(\pi/2)$ can be implemented in a time $t_{1,2}=\frac{\pi}{2\abs{g_{1,2}}}$.  Then, iteratively reducing $G_{0}$ to edges $(j,j+1)$ up to $(N-2,N-1)$ allows for successively implementing SWAP gates on adjacent vertices, which  permutes the vertices according to $12\cdots N \to 23\cdots (N-1)1N$. This takes $\tau=\frac{\pi}{2}\sum_{j=1}^{N-2}\frac{1}{\abs{g_{j,j+1}}}$ amount of time. We proceed by reducing $G_{0}$ to the single edge $(N-1,N)$ and implement the operation $S_{N-1,N}(\alpha)$, which takes $t_{N-1,N}=\frac{\abs{\alpha}}{\abs{g_{N-1,N}}}$ amount of time.  
Iteratively restoring the order of the vertices by performing $(N-2)$ SWAP operations so that $S_{1,N}(\alpha)$ is effectively implemented takes time $\tau$. Thus, the gate $S_{1,N}(\alpha)$ can be implemented in a time 
\begin{align}
\label{eq:timesequence}
t_{1,N}&=t_{N-1,N}+2\tau \nonumber \\
&=\frac{|\alpha|}{|g_{N-1,N}|}+\pi\sum_{j=1}^{N-2}\frac{1}{|g_{j,j+1}|}. 
\end{align}
Clearly, $N$ is upper bounded by the total number of vertices present in $G_{0}$, which is given by the dimension $d$ of the quantum system. Introducing the smallest edge weight present in $G_{0}$ as $g_{\text{min}}=\min_{(n,m)\in E_{0}}\{|g_{n,m}|\}$ we find that the time $t_{n,m}$ to implement a generic $S_{n,m}(\alpha)$ with $(n,m)\in E_{K}$ is therefore upper bounded by \eqref{eq:upperboundtnm}. 

\end{proof}
We remark here that the sequence of SWAP operations used to obtain \eqref{eq:timesequence} is not necessarily time optimal. However, minimizing \eqref{eq:timesequence} over all paths connecting vertices $n$ and $m$ yields the tightest version of the obtained bound.

Instead of associating $G_{0}$ with coupled states described by $B_{n,m}$, we can also consider the case where $G_{0}$ describes a $n$-qubit network (i.e., here qubits represent vertices labeled by $i$ and edges labeled by $(i,j)$  are given by two body interaction terms) described by the drift Hamiltonian  
\begin{align}
\label{eq:driftqubit}
H_{0}=\sum_{\mathclap{\substack{i\in V, \\  \alpha\in\{x,y,z\}}}}\omega_{\alpha}^{(i)}\sigma_{\alpha}^{(i)}+\sum_{\mathclap{\substack{(i, j)\in E,\\\alpha,\beta\in\{x,y,z\}}}} g_{\alpha,\beta}^{(i,j)}\sigma_{\alpha}^{(i)}\sigma_{\beta}^{(j)},
\end{align}
where $\sigma_{\alpha}^{(i)}$ denote the Pauli spin operators acting only non-trivially on the $i$th qubit, and $\omega_{\alpha}^{(i)}$ and $g_{\alpha,\beta}^{(i,j)}$ are energy splittings and coupling constants, respectively. If each qubit can be addressed with two orthogonal control fields described by the set of controls $\mathcal C=\{\sigma_{x}^{(i)},\sigma_{y}^{(i)}\}_{i=1}^{n}$, as shown in  \cite{arenz2018controlling} the time $t_{i,j}^{\text{cnot}}$ to implement a CNOT gate on qubits $i$ and $j$ is upper bounded by,     
\begin{align}
\label{eq:timeCNOT}
t_{i,j}^{\text{cnot}}\leq \frac{\pi}{g_{\text{min}}} \left(\frac{4\,\text{dist}(i,j)-3}{4} \right).
\end{align}  
Here $g_{\text{min}}$ is the smallest non-zero coupling constant present in the Hamiltonian \eqref{eq:driftqubit} and $\text{dist}(i,j)$ denotes the geodesic path distance  between two qubits $i$ and $j$ given as the smallest number of edges in a path connecting the two considered qubits, noting that  $\text{dist}(i,j)\leq n-1$.

\subsection{Upper bounding the time to implement generic unitary operations}
In order to understand how to upper bound the evolution time of an arbitrary unitary, we can decompose it into elementary interactions   (\ref{eq:intunitary}) and local controls. In general, the number of terms in such decomposition is hard to characterize \cite{d2002uniform}. For the special case of decomposing an element of $U\in \text{U}(d)$ we can however use the proof of decompositions of unitaries given in \cite{nielsen2002quantum}, which shows that \[U=V_1 V_2 \ldots V_k,\] where $k\le d(d-1)/2,$ and each $V_k$ acts nontrivially only on two specific levels $n_k,m_k$ (and is therefore isomorphic to an element of $\text{U}(2)$). By the Euler decomposition, we can furthermore decompose such an element into a gate sequence including rotations around $z$, which can be implemented through the controls $\mathcal C$ instantaneously, and into a rotation around $x$, which by Eq. (\ref{eq:upperboundtnm}) can maximally take a time of $\frac{\pi d}{g_{min}}$. We therefore obtain our final result, which we summarize in the following Theorem. 
\begin{theorem}\label{maintheorem} \textit{Let the graph associated with the drift Hamiltonian \eqref{eq:drifttrans} be connected and denote by $g_{\text{min}}$ the smallest edge weight. Then for the control system \eqref{eq:controlsystem} the time $T$ to implement a unitary operation $U_{g}\in\text{U}(d)$ is upper bounded by}
\begin{align} 
\label{eq:14}
T\le \frac{\pi d^2 (d-1)}{2g_{min}}.
\end{align}
\end{theorem}

Remarkably, in contrast to the upper bound obtained in \cite{arenz2018controlling}, the derived bound \eqref{eq:14} is independent of the target unitary transformation $U_{g}$ as well as the accuracy of implementing $U_{g}$. 

If we consider again a $n$-qubit network described by the drift Hamiltonian \eqref{eq:driftqubit} and denote by $N_{\text{CNOT}}(U_{g})$ the number of CNOT gates needed to create a specific gate $U_{g}\in \text{SU}(2^{n})$ by locally controlling each qubit, according to \eqref{eq:timeCNOT} the time $T(U_{g})$ needed to create $U_{g}$ is then upper bounded by 
\begin{align}
T(U_{g})\leq \frac{\pi (4n-7)}{4 g_{\text{min}}}   N_{\text{CNOT}}(U_{g}).
\end{align}
However, note that for creating every $U_{g}\in \text{SU}(2^{n})$, the number of CNOT gates needed must scale exponentially in $n$, whereas the prefactors have been successively improved in the last decade \cite{barenco1995elementary, mottonen2004quantum, shende2006synthesis}.

\section{\label{sec:tightness} Tightness of the Bounds}
In order to analyze the tightness of the obtained bounds, we compare the bounds \eqref{eq:upperboundtnm} and \eqref{eq:14} to previously derived lower bounds \cite{Lee2018}, as well as to minimum gate times obtained from numerical gate optimization using the GRAPE algorithm \cite{grape}, which is included in the Python package QuTip \cite{qutip}.
Similar to the method utilized in \cite{Lee2018}, a population binary search algorithm is run over $T$ until the gate error is smaller than $10^{-4}$. 

\subsection{\label{sec:nlevel1} d-level System}
We first consider a quantum system consisting of $d$ energy levels interacting in a tight-binding like manner described by the drift Hamiltonian 
\begin{align} 
\label{eq:33}
H_{0}=J\sum_{j=1}^{d-1}(\ket{j}\bra{j+1}+\ket{j+1}\bra{j}),
\end{align}
where $J$ is the coupling strength chosen to be $J=1/\sqrt{2(d-1)}$ so that $\Vert H_{0}\Vert=1$ with $\Vert \cdot \Vert$ being the Hilbert Schmidt norm. 
We assume that the energy levels $\ket{j}$ can be controlled arbitrarily so that the set of controls is given by $\mathcal C=\{\ket{j}\bra{j}\}_{j=1}^{d}$. The goal is to implement a SWAP operation (i.e., $U_{g}=S_{1,d}(\pi/2)$ ) between the first and the $d$th level. According to Lemma \eqref{th:11}, the time $t_{1,d}^{\text{swap}}$ needed to implement $U_{g}$ is upper bounded by 
\begin{align}
\label{nlevelupper}
t_{1,d}^{\text{swap}}\leq \frac{\pi}{2}(2d-3)\sqrt{2(d-1)}.
\end{align}
Based on the results in \cite{Lee2018} with further details found in Appendix \eqref{ap1} we can also lower bound $t_{1,d}^{\text{swap}}$ by 
\begin{align}
\label{nlevellower}
\sqrt{2}(d-1)\leq t_{1,d}^{\text{swap}}. 
\end{align} 
In figure 1 we plot the upper bound (black curve) and the lower bound (blue curve), as well as the minimum time $T$ needed to implement the SWAP operation $S_{1,d}(\pi/2)$ obtained from numerical gate optimization using GRAPE (green curve) as a function of the number of levels $d$. 

\begin{figure}
\centering 
\includegraphics[width=0.60\columnwidth]{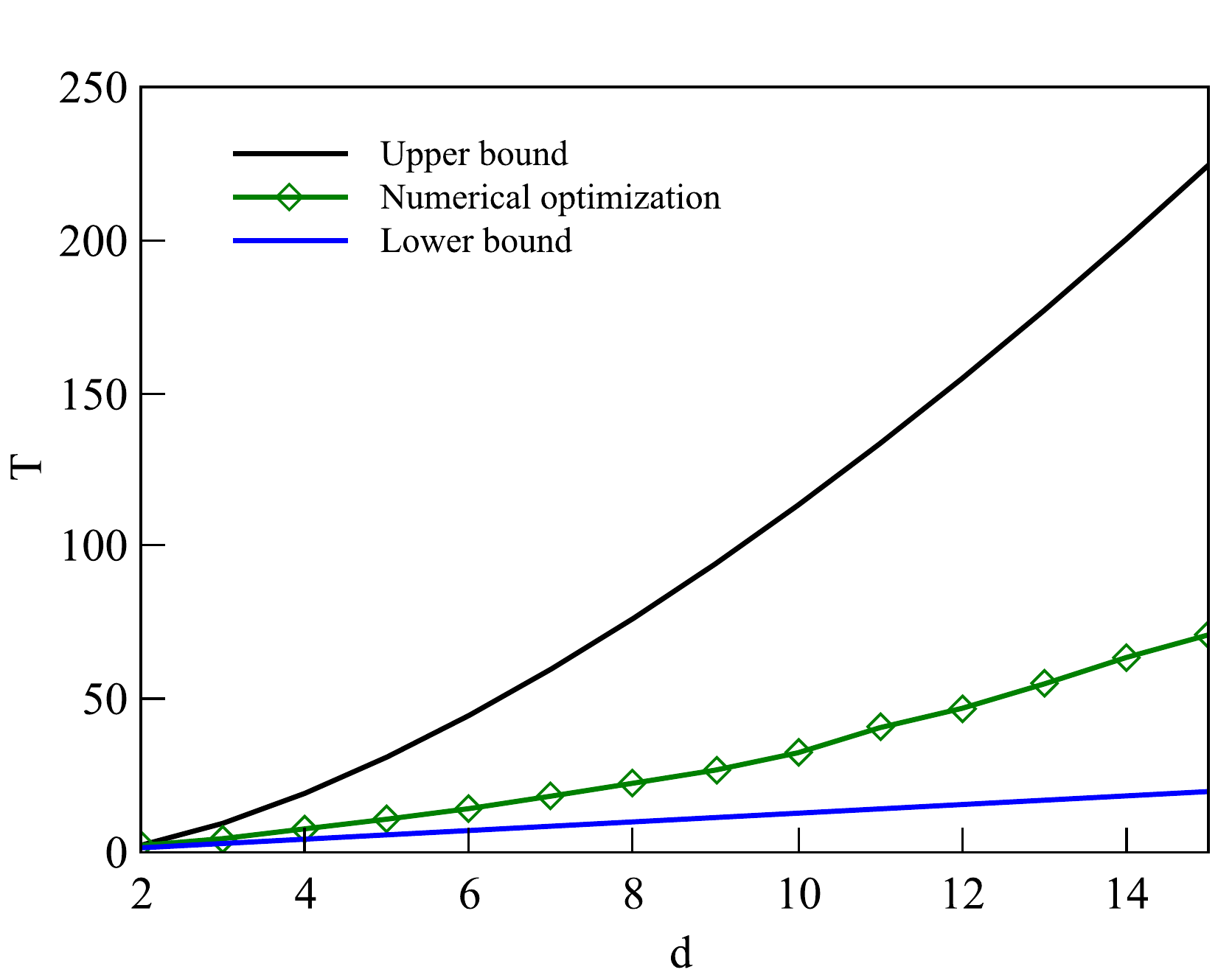}
\caption{Comparison of the time $T$ as a function of $d$ to implement a SWAP operation between the 1st and the $d$th level for the d-level system \eqref{eq:33} obtained from numerical gate optimization using GRAPE, with the upper bound (black) given in \eqref{nlevelupper} and the lower bound (blue) given in \eqref{nlevellower}.}
\end{figure}

First we observe that the time obtained from numerical optimization lies between the upper and lower bounds, as expected.
Furthermore, due to the fact that the upper bound in this system scales as $d^{3/2}$ while the lower bound scales as $d$, in terms of assessing the tightness of the bound with regard to scaling, the scaling of the upper bound deviates from the true scaling by at most $d^{1/2}$, which is sub-linear.

\subsection{\label{sec:randsingle} Random graphs}
We proceed by considering random drift Hamiltonians that correspond to random connected graphs. Throughout this section the couplings $|g_{n,m}|$ are chosen to be uniformly random in the interval $[1,2]$ so that $g_{\text{min}}=1$ and we study the validity of the bounds \eqref{eq:upperboundtnm} and \eqref{eq:14} for quantum systems of dimension $d\in[2,6]$, noting that for a fixed $d\in[2,6]$ we have $\{1,2,6,21,112\}$ distinct connected graph.  

\subsubsection{Single edge operations} We begin by studying the tightness of the bound \eqref{eq:upperboundtnm} by considering random  single edge operations $S_{n,m}(\alpha)$ by picking $\alpha$ uniformly random in $[-\frac{\pi}{2},\frac{\pi}{2}]$ and 10 random edges amongst all $\binom{d}{2}$ of the complete graph. According to \eqref{eq:upperboundtnm} for $g_{\text{min}}=1$ the time $t_{n,m}$ to implement such a random $S_{n,m}$ is then upper bounded by 
\begin{align}
\label{singleedgebound}
t_{n,m}\leq \pi\left(d-\frac{3}{2}\right),
\end{align}
which is shown as a function of $d$ (black curve) in figure 2. The green and the orange curves correspond to the times to implement $S_{n,m}$ obtained from numerical gate optimization, where we plotted the average (green) taken over 10 randomly chosen $S_{n,m}$ and the number of different distinct connected graphs for $d\in[2,6]$ (i.e., the average was taken over $\{10,20,60,20,1120\}$ different runs), whereas  the orange curve shows the maximum value.  

From figure 2 we see that the upper bound is indeed above the maximum amongst random GRAPE runs. As expected, while the average over random single-edge operations tends to be lower than the upper bound, the maximum remains relatively close to the bound \eqref{singleedgebound}. Remarkably, all three data sets are linear in $d$, showing that the upper bound captures the scaling of the minimum time reasonably well.
 
\begin{figure}[H]
\begin{center}
\includegraphics[width=0.60\columnwidth]{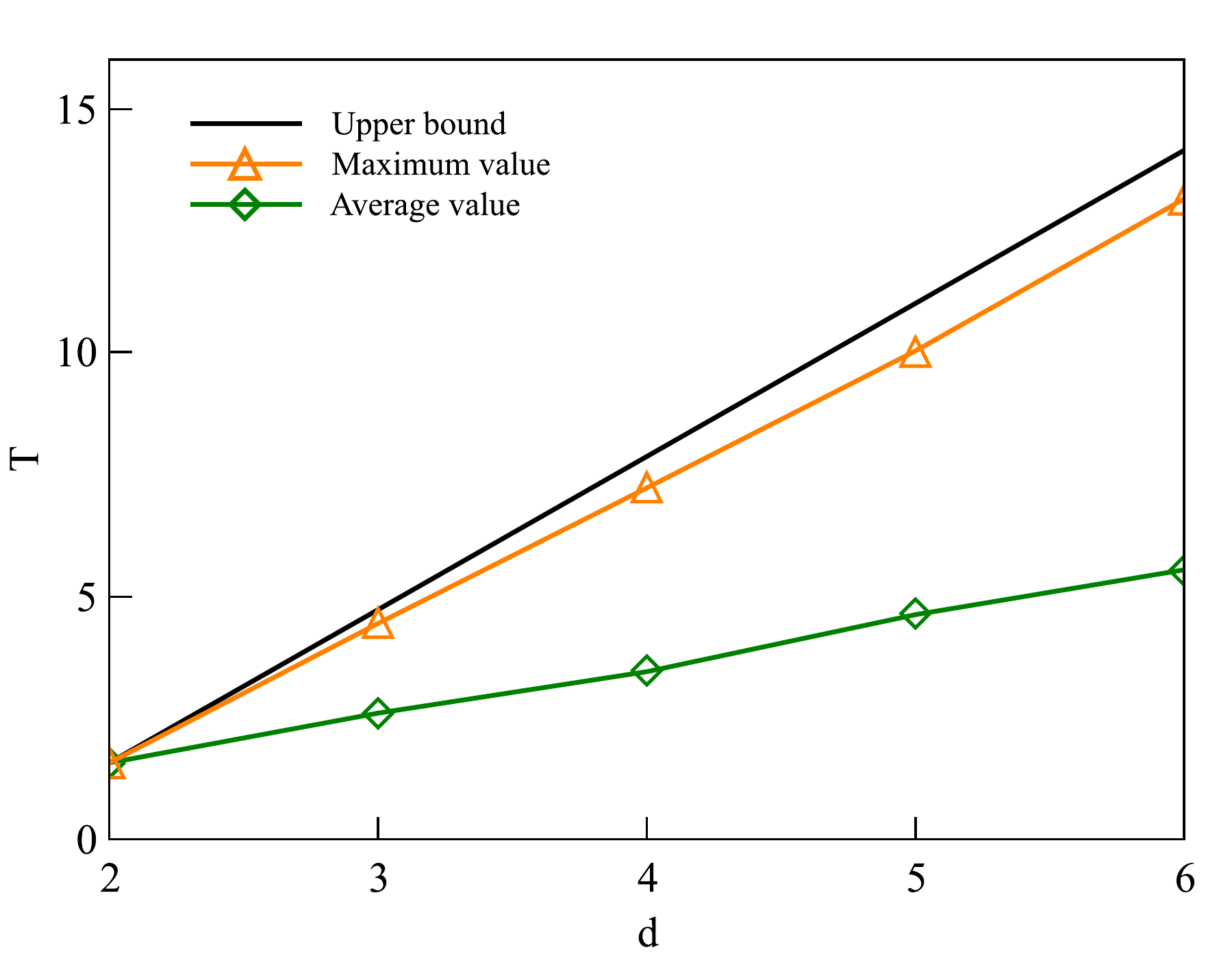}
\caption{ Comparison of the time $T$ to implement a random single edge operation given by \eqref{eq:intunitary} on a quantum system of dimension $d$ described by a randomly chosen connected graph obtained from numerical gate optimization using GRAPE, with the upper bound (black) given in \eqref{singleedgebound}. The green curves shows the average over $\{10,20,60,20,1120\}$ with each value corresponding to a fixed $d\in[2,6]$ and the orange curve shows the maximum value. Further details can be found in the main body of the manuscript.}
\end{center}
\end{figure}

\subsubsection{\label{sec:randgen} General Behavior}
Finally, we test the tightness of the upper bound \eqref{eq:14} for the time $T$ to implement generic unitary operations $U_{g}\in\text{U}(d)$. For $g_{\text{min}}=1$ we have, 
\begin{align}\label{generalbound}
T\leq \frac{\pi}{2}d^{2}(d-1).
\end{align}
\begin{figure}
\begin{center}
\includegraphics[width=0.60\columnwidth]{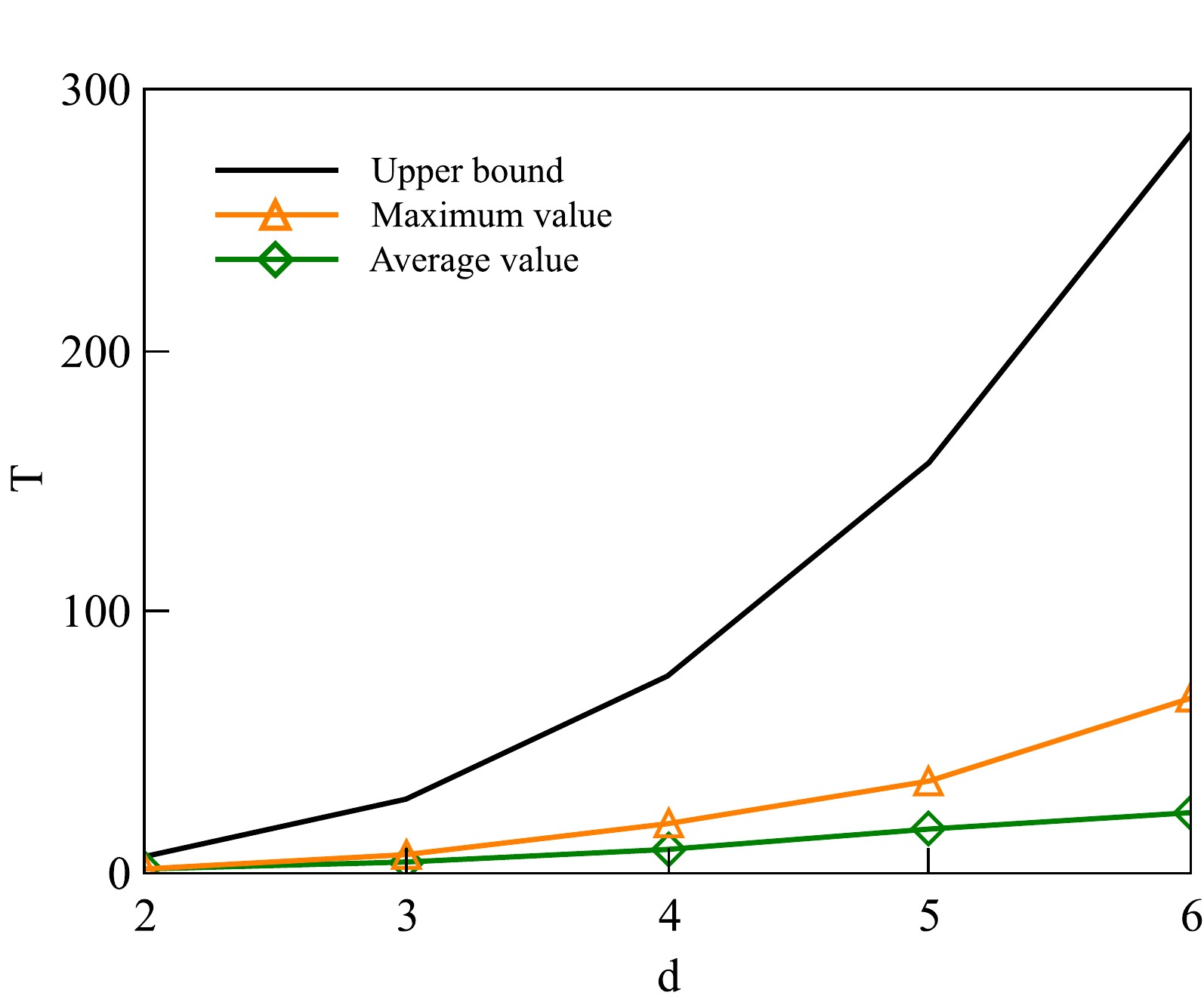}
\caption{ Comparison of the time $T$ to implement a random unitary operation on a quantum system of dimension $d$ described by a randomly chosen connected graph obtained from numerical gate optimization using GRAPE, with the upper bound (black) given in \eqref{eq:14}. The green curves shows the average over $\{10,20,60,20,1120\}$ with each value corresponding to a fixed $d\in[2,6]$ and the orange curve shows the maximum value. Further details can be found in the main body of the manuscript.}
\end{center}
\end{figure}
As in the single-edge operation case, for each dimension $d$ we run the trials on every distinctly homeomorphic connected graph. We construct random unitary operations $U_{g}=\exp(-iH)$ by picking a random hermitian matrix $H$. 

From the results in figure 3 we again note that the upper bound is indeed above the maximum amongst random GRAPE runs. 
Furthermore, we can see a "similar" polynomial scaling for the maximum and the average times. 
\section{\label{sec:disc} Conclusions}
We have derived an upper bound for the time $T$ to implement a generic unitary transformation on a quantum system in which the diagonal element (in a given basis) can be controlled arbitrarily. This was achieved by first describing the considered system as an undirected and connected graph, followed by showing that edges of the graph can be removed without a time cost using a decoupling sequence generated by the controlled diagonal elements. Afterwards we showed in Lemma 1 that the time to implement unitary operations generated by generic edges of the complete graph scales at most linearly in the dimension of the system. Consequently, every unitary transformation can be implemented in a time at most $\mathcal O(d^{3})$, which was summarized in Theorem 1. It is interesting to note that the corresponding upper bound on $T$ given in \eqref{eq:14} is independent of the target unitary transformation and the accuracy the unitary transformation is implemented. 

Based on the results in \cite{arenz2018controlling} we also derived an upper bound for the time to create a unitary transformation in a qubit network in which each qubit can be locally controlled in terms of the number of CNOT gates needed to create the unitary transformation.

By considering examples we numerically studied the tightness of the obtained bounds and found that the bounds capture the system size dependence of $T$ remarkably well. 

One of the key assumptions in this works was to assume that the control fields are unconstrained so that interactions can be removed instantaneously through a decoupling sequence. Recently, however, it was shown \cite{lloyd2019efficient} that under some assumptions on the drift Hamiltonian even limited control fields can yield a desirable scaling of the minimum time for implementing unitary transformations. It would be interesting to combine the approaches used in \cite{lloyd2019efficient} with the results obtained here to characterize the number and the type of controls needed to efficiently control quantum systems.  \\

\textit{Acknowledgements} -- J.L. acknowledges financial support from the Program in Plasma Science and Technology at Princeton University. C. A. acknowledges funding from the DOE (DE-FG02-02ER15344) and H. R. acknowledges funding from the NSF (grant  CHE-1763198).       

\bibliographystyle{unsrt}
\bibliography{references.bib}
%\newpage
\appendix
\section{A Lower Bound for the d-level system}\label{ap1}
Using the results obtained in \cite{Lee2018} we show here that for the controlled d-Level system described by 
\begin{align}
H(t)=J\sum_{j=1}^{d-1} \left(\ket{j}\bra{j+1}+\ket{j+1}\bra{j} \right)+\sum_{j=1}^{d}f_{j}(t)\ket{j}\bra{j}, 
\end{align}
where $J=\frac{1}{\sqrt{2(d-1)}}$, the time $T$ to implement a SWAP operation
\begin{align}
U_{g}=\exp\left(-i\frac{\pi}{2}\left(\ket{1}\bra{d}+\ket{d}\bra{1}\right)\right),
\end{align}
between the 1st and the $d$th levels is lower bounded by 
\begin{align}
\label{eq:boundswap}
\sqrt{2}(d-1)\leq T. 
\end{align}
From \cite{Lee2018} we have that in general for a control system of the form
\begin{align}
H(t)=H_{0}+\sum_{k=1}^{n}f_{k}(t)H_{k}, 
\end{align}
evolving on the unitary group $\text{U}(d)$ the time $T$ to implement a generic $U_{g}\in \text{U}(d)$ is lower bounded by 
\begin{align}
\label{eq:lower}
\max_{V\in \bigcap_{k}\text{Stab}(iH_{k})}\frac{\Vert [U_{g},V]\Vert }{\Vert [H_{0},V]\Vert}\leq T,
\end{align}
valid for any unitarily invariant norm where $\text{Stab}(iH_{k})$ denotes the stabilizer of $iH_{k}$ defined as $\text{Stab}(x)=\{U\in\text{U}(d)\,|\,U^{\dagger}xU=x\}$ for some $x\in\mathfrak u(d)$.

For the d-level control system the intersection of the stabilizers is given by  $V=\text{diag}(e^{i\theta_{1}},\ldots,e^{i\theta_{d}})$ so that using the Hilbert Schmidt norm defined as $\Vert A\Vert=\sqrt{\text{Tr}\{A^{\dagger}A\}}$ yields 
\begin{align*}
\frac{\Vert [U_{g},V]\Vert^{2} }{\Vert [H_{0},V]\Vert^{2}}= \frac{1}{J^{2}}\cdot \frac{1-\cos{(\theta_{1}-\theta_{d})}}{d-1-(\cos{(\theta_{1}-\theta_{2})}+\ldots+\cos{(\theta_{d}-\theta_{d-1})})}.
\end{align*}
By defining $x_{i}\equiv\theta_{i}-\theta_{i+1}$ and 
\begin{align}
\label{eq:righthands}
S(x_{1},\cdots, x_{d-1})\equiv \frac{1-\cos(\sum_{i=1}^{d-1}x_{i})}{1-\frac{1}{d-1}\sum_{i=1}^{d-1}\cos(x_{i})},
\end{align} 
the maximization in \eqref{eq:lower} is then equivalent to maximizing $S(x_{1},\cdots,x_{d-1})$ over all $x_{i}$.
We claim that this quantity is maximized when $x_{i}=0$ for all $i$, which implies that $V$ is given by the identity (up to a global phase). To see this we set $x_{i}=x$ for all $i$ and take the limit,  
\begin{align*}
\lim_{x\rightarrow 0}S(x) = (d-1)^{2},  
\end{align*}
which is indeed the maximum. We therefore have $2(d-1)^{2}\leq T^{2}$, which yields the  desired result \eqref{eq:boundswap}. 

\end{document}